\documentclass[sigconf]{acmart}

\makeatletter
\renewcommand\@formatdoi[1]{\ignorespaces}
\makeatother

\usepackage{fancyhdr}
\usepackage{booktabs} 
\usepackage{amssymb}
\usepackage{amsmath}
\usepackage{graphicx}
\usepackage[tight,footnotesize]{subfigure}
\usepackage{multirow}
\usepackage{listings}
\usepackage{color}
\usepackage{caption}
\usepackage{algorithm}
\usepackage{algorithmicx}
\usepackage[noend]{algpseudocode}
\usepackage{xspace}
\usepackage{url}
\usepackage{breakurl}
\usepackage{flushend}
\usepackage{enumerate}

\newcommand{\sysname}{\textsc{\small{FraudTrap}}\xspace}
\newcommand{\sysnamep}{\textsc{\small{FraudTrap+}}\xspace}
\newcommand{\sysnames}{\textsc{\small{FraudTrap*}}\xspace}

\newcommand{\sscore}{$\mathcal{S}_{ij}$ }
\newcommand{\cij}{$C_{ij}$ }

\newcommand{\seti}{$\mathbf{I}_i$ }
\newcommand{\para}[1]{\xspace \smallskip \noindent\textbf{#1} \ }
\newcommand{\edge}{$\boldmath{\varepsilon}$ }
\newcommand{\subg}{$\mathcal{G}$\xspace}
\newcommand{\obg}{$\mathbf{G}$\xspace}
\newcommand{\obj}{$\mathbf{M}$\xspace}
\newcommand{\user}{$\mathbf{N}$\xspace}

\newcommand{\fscore}{$\mathcal{F}$-score\xspace}

\newcommand{\cscore}{$\mathbf{C}$-score\xspace}

\makeatletter
\newcounter{protocol}

\makeatother

\lstdefinestyle{mystyle}{
	keywordstyle=\color{magenta},
	basicstyle=\ttfamily\footnotesize,
	breakatwhitespace=false,         
	breaklines=true,                 
	captionpos=b,                    
	keepspaces=true,                 
	numbers=left,                    
	numbersep=5pt,                  
	showspaces=false,                
	showstringspaces=false,
	showtabs=flase,                  
	tabsize=2,
}

\acmConference[WWW'19]{}{2019}{}

\acmYear{2019}
\copyrightyear{2019}

\author{Yikun Ban}
\affiliation{%
  \institution{Peking University}
  \streetaddress{P.O. Box 1212}
  \postcode{100084}
}

\author{Jiao Sun}
\affiliation{%
	\institution{Tsinghua University}
	\streetaddress{P.O. Box 1212}
	\postcode{}
}

\author{Xin Liu}
\affiliation{%
	\institution{Tsinghua University}
	\streetaddress{P.O. Box 1212}
	\postcode{100084}
}

\author{Ling Huang}
\orcid{1234-5678-9012}
\affiliation{%
	\institution{Fintec.ai}
	\streetaddress{P.O. Box 1212}
	\postcode{100084}
}

\author{Yitao Duan}
\orcid{1234-5678-9012}
\affiliation{%
	\institution{Netease Youdao Inc.}
	\streetaddress{P.O. Box 1212}
	\postcode{100084}
}

\author{Wei Xu}
\affiliation{%
	\institution{Tsinghua University}
	\streetaddress{P.O. Box 1212}
	\postcode{100084}
}

\begin{document}

\title{FraudTrap: Catching Loosely Synchronized Behavior in Face of Camouflage}

\begin{abstract}

The problem of online fraud detection can often be formulated as mining a bipartite graph of users and objects for suspicious patterns. The edges in the bipartite graph represent the interactions between users and objects (e.g., reviewing or following). However, smart fraudsters use sophisticated strategies to influence the ranking algorithms used by existing methods. Based on these considerations, we propose \sysname, a fraud detection system that addresses the problem from a new angle.  Unlike existing solutions, \sysname works on the object similarity graph (\textbf{OSG}) inferred from the original bipartite graph. The approach has several advantages. First, it effectively catches the loosely synchronized behavior in face of different types of camouflage. Second, it has two operating modes: unsupervised mode and semi-supervised mode, which are naturally incorporated when partially labeled data is available to further improve the performance. Third, all algorithms we design have the near-liner time complexities and apply on large scale real-world datasets. Aiming at each characteristics of \sysname, we design corresponding experiments that show \sysname outperforms other state-of-the-art methods on eleven real-world datasets.
\end{abstract}

\maketitle

%
%

\section{Introduction}
Fraud has severely detrimental impacts on the business of social networks and other web online applications \cite{CrimeReport}. A user can become a fake celebrity by purchasing ``zombie followers'' on Twitter. A merchant can boost his reputation through fake reviews on Amazon. This phenomenon also conspicuously exists on Facebook, Yelp and TripAdvisor, etc. In all the cases, fraudsters try to manipulate the platform's ranking mechanism by faking interactions between the fake accounts they control (\emph{fraud users}) and the target customers (\emph{fraud objects}). 

 These scenarios are often formulated as a bipartite graph of objects and users. We define an object as the target a user could interact with on a platform. Depending on the application, an object can be a followee, a product or a page. An edge corresponds to the interaction from a user to the object (e.g., reviewing or following). Detecting fraud in the bipartite graph has been explored by many methods. Since fraudsters rely on fraudulent user accounts, which are often limited in number, to create fraudulent edges for fraud objects' gain \cite{CATCHSYNC}, previous methods are mainly based on two observations: (1) fraud groups tend to form dense subgraphs in the bipartite graph (high-density signal) , and/or (2) the subgraphs induced by fraud groups have unusually surprising connectivity structure (structure signal). These methods mine the bipartite graph directly for dense subgraphs or rare structure patterns. Their performances vary in real-world datasets. 

Unfortunately, smart fraudsters use more sophisticated strategies to avoid such patterns. First, by multiplexing a larger pool of fraud users, a fraudster can effectively reduce the density of the subgraph induced by a fraud group. This is called \emph{loosely} synchronized behavior and leads to the limited performance of the methods\cite{FRAUDAR, CROSSSPOT, KCORE, SPOKEN, NETPROBE} depending on the high-density signal.  Another commonly used technique is to create edges pointing to normal objects to disguise fraud users as normal ones. This strategy,  often called \emph{camouflage}, alters the connectivity structure of the bipartite graph and weakens the effectiveness of many approaches targeting such structure, such as HITS\cite{CATCHSYNC, ECommerce}, and \emph{belief propagation (BP)}\cite{FRAUDEAGLE, NETPROBE}. Fig.~\ref{fig:attack} illustrates these two strategies.

\begin{figure}
\includegraphics[width = \columnwidth]{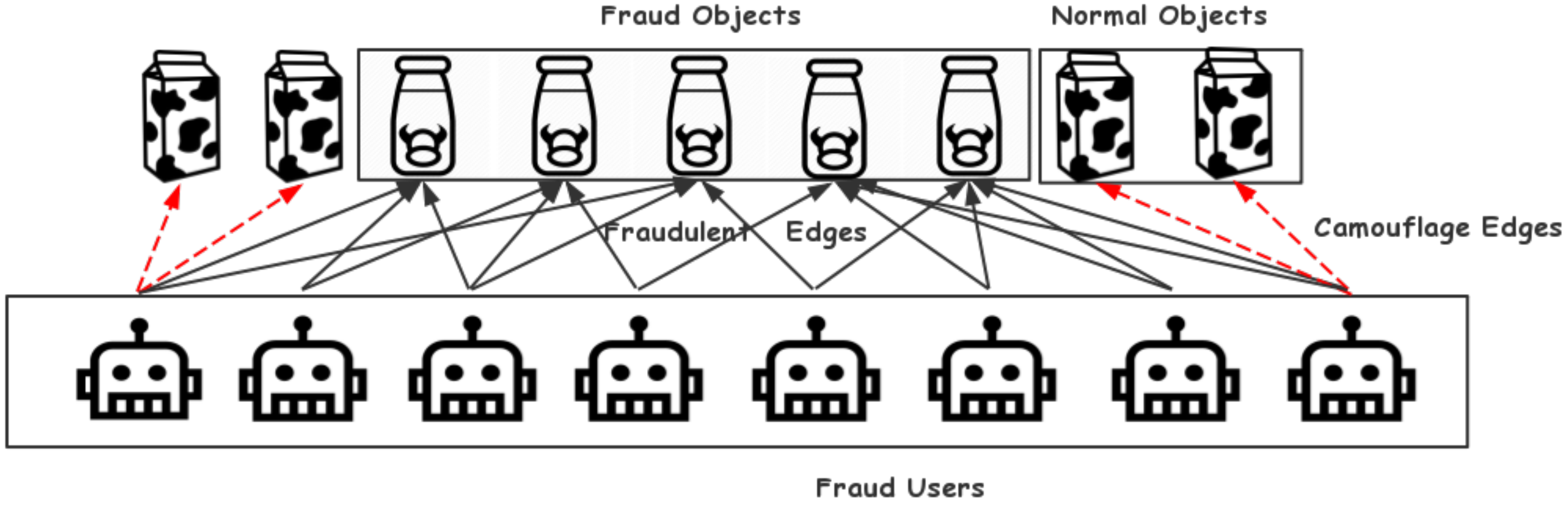}
\vspace{-2.0em}
\caption{(1) Loosely synchronized behavior: fraudsters increase the number of fraud users and multiplex them. (2) Camouflage: fraud users create edges to normal objects.} 
\label{fig:attack}
\vspace{-2em}

\end{figure}

The problem of fraud detection can also be handled using supervised or semi-supervised approaches when (partially) labeled data are available. \cite{ADOA, SYBILINFER} provide better performance using a subset of labeled frauds. 
 \cite{Abdulhayoglu2017HinDroid, 2010Uncovering, Egele2017Towards} build machine learning classifiers to detect anomalies. These approaches, however, have a number of limitations. Firstly, it is often very difficult to obtain enough labeled data in fraud detection due to the scale of the problem and the cost of investigation. Secondly, they require great effort in feature engineering which is tedious and demands high expertise level. Thirdly, they often fail to detect new fraud patterns. Finally, even though some labeled data can provide potentially valuable information for fraud detection, it is not straightforward to incorporate them into existing unsupervised or semi-supervised solutions such as \cite{ADOA, SYBILINFER}. 

In this paper, we propose \sysname, a graph-based fraud detection algorithm that overcomes these limitations with a novel change of the target of analysis. Instead of mining the bipartite graph directly, \sysname analyzes the Object Similarity Graph (\textbf{OSG}) that is derived from the original bipartite graph. There are two main advantages of our design: (1) fraud objects exhibit more similar behavior patterns since fraud objects are difficult to gain edges from normal users, which endows \sysname with inherent camouflage-resistance (Sec.~\ref{sec:analysis}); (2) since the number of objects is typically smaller than the number of users~\cite{amazon_data}, working with OSG reduces computation cost while guaranteeing the effectiveness. In addition, although \sysname works well without any labels, we can easily switch to a semi-supervised mode and improve performance with partial labels.


In summary, our main contributions include:

1) \textbf{[Metric $C$]}. We build Object Similarity Graph (OSG) by a novel similarity metric, $C$-score, which transforms the sparse subgraphs induced by fraud groups in the bipartite graph into the much denser subgraphs in OSG, by merging information from unlabeled and labeled(if available) data.

2) \textbf{[Algorithm LPA-TK]}. We propose a similarity-based clustering algorithm, LPA-TK, that perfectly fits in OSG and outperforms the baseline (LPA) in face of noise edges (camouflage). 

3) \textbf{[Metric $\mathcal{F}$]}. Given candidate groups returned by $C$ + LPA-TK, we propose an interpretable suspiciousness metric, $\mathcal{F}$-score, meeting the all basic ``axioms'' proposed in \cite{CROSSSPOT}.

4) \textbf{[Effectiveness]}. Our method \sysname ($C$ + LPA-TK + $\mathcal{F}$) can operate in two modes: unsupervised and semi-supervised. The unsupervised mode outperforms other state-of-the-art methods for catching synchronized behavior in face of camouflage. Semi-supervised mode naturally takes advantage of partially labeled data to further improve the performance.

\section{Related work}

To maximize their financial gains, fraudsters have to share or multiplex certain resources (e.g., phone numbers, devices).  To achieve the ``economy of scale'', fraudsters often use many fraudulent user accounts~\footnote{To be succinct, we use \emph{fraud users} to refer to these accounts. } to conduct the same fraud \cite{SYNCHROTRAP, COPYCATCH}. As a result, fraud users inevitably exhibit synchronized behavior on certain features, be it phone prefixes, or IP subnets. Group-based approaches that detect frauds by identifying such synchrony are surpassing content-based approaches (e.g., \cite{SPAMREVIEW}) as the most effective anti-fraud solutions. There are three classes of methods. 

\para{Unsupervised.} Unsupervised methods achieve various performance on fraud detection.  There are two types of unsupervised detection methods in the literature.  

The first type is based on high-density subgraphs formed by fraud groups. Mining dense subgraphs in the bipartite graph \cite{KCORE, SPOKEN, DCUBE} is effective to detect the fraud group of users and objects connected by a massive number of edges. Fraudar \cite{FRAUDAR} tries to find a subgraph with the maximal average degree using a greedy algorithm.  CrossSpot \cite{CROSSSPOT} focuses on detecting dense blocks in a multi-dimensional tensor and gives several basic axioms that a suspiciousness metric should meet.  People have also adopted singular-value decomposition (SVD) to capture abnormal dense user blocks~\cite{INFERRING,FBOX}. However, fraudsters can easily evade detection by reducing the synchrony in their actions (details in Sec.~\ref{sec:camouflage}). 

The second type is based on rare subgraph structures of fraud groups. Such structures may include the sudden creation of massive edges to an object, etc.BP \cite{FRAUDEAGLE, NETPROBE} and HITS~\cite{COMBATING, CATCHSYNC,UNDERSTAN} intend to catch such signals in the bipartite graph.  FraudEagle \cite{FRAUDEAGLE} uses the loopy belief propagation to assign labels to the nodes in the network represented by Markov Random Field (MRF).  \cite{Shah2017EdgeCentric} ranks abnormality of nodes based on the edge-attribute behavior pattern by leveraging minimum description length. \cite{Kumar2017FairJudge, Hooi2015BIRDNEST} use Bayesian approaches to address the rating-fraud problem. SynchroTrap \cite{SYNCHROTRAP} works on the user similarity graph. In all the cases, it is relatively easy for fraudsters to manipulate edges from fraud users to conceal such structural patterns (details in Sec. \ref{sec:camouflage}). The common requirement of parameter tuning is also problematic in practice, as the distribution of fraudsters changes often. 

\begin{table}
\footnotesize
\caption{\sysname v.s. existing methods}
\vspace{-0.7em}
\begin{tabular}{c|ccccccll|l}
\toprule
  & \rotatebox{90}{Fraudar\cite{FRAUDAR}} &\rotatebox{90}{Spoken \cite{SPOKEN} }& \rotatebox{90}{CopyCatch\cite{COPYCATCH} }&\rotatebox{90}{CatchSync \cite{CATCHSYNC}} &\rotatebox{90}{CrossSpot \cite{CROSSSPOT}}&  \rotatebox{90}{Fbox\cite{FBOX}} & \rotatebox{90}{FraudEagle\cite{FRAUDEAGLE} }&\rotatebox{90}{M-zoom\cite{MZOOM}} &\rotatebox{90}{\sysname} \\
\midrule
Loose synchrony?&$\times$&$\times$&$\times$&$\surd$&$\times$&$\times$&$\times$&$\times$ &$\surd$\\
Camou-resistant?& $\surd$ & $\times$ & $\surd$ & $\times$ & ? & $\times$ &$\times$&$\surd$ & $\surd$ \\
Side information?& $\times$ & $\surd$ & $\surd$ & $\times$ & $\surd$ & $\times$ &$\surd$&$\surd$& $\surd$ \\
Semi-supervised?& $\times$ & $\times$ & $\times$ & $\times$ & $\times$ & $\times$ &$\times$&$\times$& $\surd$ \\
\bottomrule
\end{tabular}
\centering
\vspace{-2em}
\end{table}

\para{(Semi-)supervised.} When partially labeled data are available, semi-supervised methods can be applied to anomaly detection. The fundamental idea is to use the graph structure to propagate known information to unknown nodes.  \cite{SEMIGRAPH, SYBILBELIEF} model graphs as MRFs and label the potential suspiciousness of each node with BP. \cite{INTEGRO, KEEPFRIENDS, SYBILINFER} use the random walk to detect Sybils. ADOA\cite{ADOA} clusters observed anomalies into $k$ clusters and classifies unlabeled data into these $k$ clusters according to both the isolation degree and similarity.  When adequate labeled data are available, people have shown success with classifiers such as multi-kernel learning\cite{Abdulhayoglu2017HinDroid}, support vector machines \cite{Tang2009Machine} and $k$-nearest neighbor \cite{KNEAR}. However, it is rare to have enough fraud labels in practice.

\section{Design Considerations}\label{sec:camouflage}

We provide details why fraudsters can easily evade existing detection, and present the key ideas of \sysname design.

\subsection{How Smart Fraudsters Evade Detection? }

\para{Reducing synchrony in fraud activities.}
One of the key signals that existing fraud detection methods rely on is the high-density of a subgraph. A naive fraud campaign may reuse some of the resources such as accounts or phone numbers, resulting in high-density subgraphs. However, experience shows that fraudsters now control larger resource pools and thus can adopt smarter strategies to reduce the synchrony by rotating the fraud users each time. 
For example, \cite{CATCHSYNC} reports that on Weibo a fraud campaign uses 3 million fraud accounts, a.k.a. \emph{zombie fans},  to follow only 20 followees (fraud objects).  Each followee gains edges from a \emph{different} subset of the followers \cite{CATCHSYNC}. The \emph{edge density} (the ratio of the number of edges to the maximum number of possible edges given its nodes) of the subgraph induced by the fraud group is only  $3.3 \times 10^{-6}$, which is very close to legit value.
This strategy, as our experiments in Sec.\ref{exp:syn} will show, effectively reduces the synchrony and deceives many subgraph-density-based methods~\cite{FRAUDAR, KCORE, SPOKEN, DCUBE, CROSSSPOT}. For example, FRAUDAR~\cite{FRAUDAR}, it is susceptible to synchrony reduction (details in Sec.\ref{exp:syn}).

\para{Adding camouflage. }
Fraudsters also try to confuse the detection algorithm by creating camouflage edges to normal objects, making the fraud users behave less unusually (Fig.\ref{fig:attack} (2)). According to \cite{FRAUDAR}, there are four types of camouflages: 1) random camouflage: adding camouflage edges to normal objects randomly; 2) biased camouflage:  creating camouflage edges to normal objects with high in-degree. 3) hijacked accounts: hijacking honest accounts to add fraudulent edges to fraud objects. 4) reverse camouflage: tricking normal users to add edges to fraud objects. 

Camouflage severely affects graph-structure-based methods\cite{FRAUDEAGLE, NETPROBE, COMBATING, CATCHSYNC, UNDERSTAN}, as fraudsters can reshape the structure without many resources. For example, our experiments in Sec.\ref{exp:syn} demonstrate that the degrees and HITS scores from Catchsync~\cite{CATCHSYNC} stops working even with a moderate number of camouflage edges. 

\subsection{Our Key Ideas }

The fundamental reason that the above two strategies succeed in deceiving existing detection methods is that they are based on analyzing the original bipartite graph.  The fraudsters can easily manipulate the graph (both the density and structure) with a large number of fraud users. Unfortunately, the current black market makes the number of fraud accounts easy to obtain.


We propose to solve the problem from a different angle. Objects that pay for fraud activities are similar because the fraudsters must use their fraud user pool to serve many objects to make profits. Thus, instead of analyzing the user-object bipartite graph directly, we work on the similarity among different objects, which we capture as an \emph{object similarity graph (OSG)} whose nodes are all objects and the edges represent the similarity among these objects. As we will show, with a carefully designed similarity score, a fraud object is much more similar with other fraud objects than normal ones and it is much harder for fraudsters to manipulate the OSG than the original bipartite graph.  This is because, in the OSG, the subgraph $\mathcal{G}$ formed by loosely synchronized behavior is much denser compared to the corresponding subgraph in the original user-object bipartite graph and the density of $\mathcal{G}$ cannot be altered by camouflage.  Figure \ref{fig:cscore} shows an intuitive example.   

Furthermore, we want to leverage the side information available in different applications instead of letting the algorithm limit the choices.  Specifically, we allow optionally including two types of information, (partial) fraud labels to offer a semi-supervised mode for the algorithm and side information of the activities, such as timestamp and star rating, etc.  As we will show, the similarity score we design is additive for both labels and extra dimensions, so it is easy to incorporate all available information into the uniform framework.

\begin{figure}[t]
\includegraphics[width = \columnwidth]{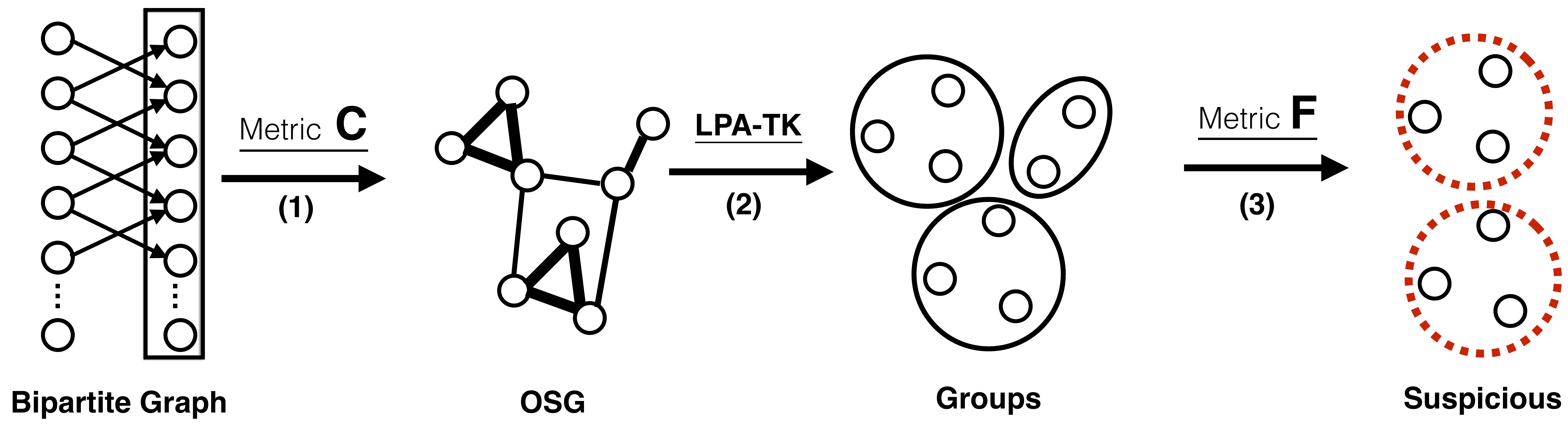}
\centering
\caption{The workflow of \sysname ($C$+LPA-TK+$\mathcal{F}$).} \label{fig:workflow}
\vspace{-2.0em}
\end{figure}

\section{Methods} \label{sec:me}

In this section, we first provide an overview of the workflow (Figure \ref{fig:workflow}), and then we detail each of the three steps of the OSG construction ($C$), clustering on OSG (LPA-TK), and spot suspicious groups ($\mathcal{F}$).  Finally, we provide the intuitions and proofs.

\para{Problem definition and workflow.}
Consider a bipartite graph $G$ of a user set \user and an object set \obj, and another bipartite graph $G^l$ formed by a subset of labeled fraud users $\mathbf{N}^l$ and the same object set \obj.  
We use an edge \edge pointing from a user to an object to represent an interaction between them, be it a follow, comment or purchase. \sysname works in three stages:

\begin{enumerate}
\item OSG construction: 
The OSG captures the object similarity, and we design a metric, \cscore, to capture the similarity between two objects based on user interactions. If $G^l$ is available, i.e., there are some labeled data, the \cscore incorporates that data too. 

\item Clustering on OSG: 
We propose a similarity-based clustering algorithm that clusters each object into a group based on its most similar neighbors on OSG.


\item Spot suspicious groups: 
Given candidate groups, it is important to use an interpretable metric to capture \emph{how suspicious an object/user group is, relative to other groups}. We design the $\mathcal{F}$-score metric for the purpose.  


\end{enumerate}

We elaborate these three stages in the rest of this section.

\begin{table}[t]
\caption{Symbols and Definitions}\label{tab1}
\vspace{-0.7em}
 \centering
\begin{tabular}{l|l}
\toprule
Symbols &  Definition\\
\midrule
\user &  The set of users, $\mathbf{N} = \{ n_1, ..., n_{|\mathbf{N}|} \}$  \\ 
$\mathbf{N}^l$ & The set of labeled fraud users, $\mathbf{N}^l \subset \mathbf{N}$\\
\obj &  The set of objects, $\mathbf{M} = \{m_1, ... , m_{|\mathbf{M}|}\} $  \\
$G$ & The bipartite graph, $G =(\mathbf{N} \cup \mathbf{M}, E)$\\
$G^l$ & The bipartite graph, $G^l =(\mathbf{N}^l \cup \mathbf{M}, E^l)$\\
$\varepsilon_{ij}$ & An edge, $\varepsilon_{ij} \in E/E^l$ and $\varepsilon_{ij} = (n_i, m_j), n_i \in \mathbf{N}/\mathbf{N}^l$\\  
$\mathbf{I}_i$ & The set of edges pointing to $m_i$, $\mathbf{I}_i \subset E$ \\
$\mathbf{I}_i^l$ & The set of edges pointing to $m_i$, $\mathbf{I}_i^l \subset E^l$ \\

\obg & Object Similarity Graph, $\mathbf{G} = (\mathbf{M}, \mathbf{E})$\\
$C_{ij}$ & Object Similarity Score, $C_{ij} \in \mathbf{E}$ and $C_{ij}= (m_i, m_j) $ \\
\subg & A subgraph of \obg, $\mathcal{G} = (\mathcal{M}, \mathcal{E})$\\ 
\bottomrule
\end{tabular}
\vspace{-2em}
\end{table}

\subsection{Stage I: OSG Construction ($C$) }\label{sec:con}

OSG captures the similarity between object pairs, and thus the first step is to define the similarity metric, $C$-score.  The $C$-score has two parts, similarity in $G$ (unlabeled) and in $G^l$ (labeled). Formally, we define the similarity score \cij between object $m_i$ and object $m_j$ as 

\begin{equation}\label{equ:cscore}
C_{ij} = \mathcal{S}_{ij} + \mathcal{S}^l_{ij},
\end{equation}
where \sscore is the similarity score calculated from the unlabeled $G$, while $\mathcal{S}_{ij}^l$ is obtained from the labeled $G^l$.

In $G$, let \seti $= \{\varepsilon_{ji}: n_j \in \mathbf{N}, (n_j, m_i)\in E \}$  be the set of edges pointing to $m_i$. Following the definition of the Jaccard similarity~\cite{JACCARD}, we define the similarity between $m_i$ and $ m_j$, \sscore, as

\begin{equation}\label{eq:c1}
\mathcal{S}_{ij} = \frac{\left | \mathbf{I}_i\cap \mathbf{I}_j \right |}{\left | \mathbf{I}_i \cup \mathbf{I}_j \right |}.
\end{equation}

In $G^l$, let $\mathbf{I}_i^l = \{\varepsilon_{ji}:n_j \in \mathbf{N}^l, (n_j, m_i)\in E^l \} $ represent the set of edges pointing to $m_i$. Then the similarity score $\mathcal{S}_{ij}^l$ between $m_i$ and $m_j$ is given by:
\begin{equation}\label{eq:c2}
\mathcal{S}_{ij}^l = \frac{\left |   \mathbf{I}_i^l\cap \mathbf{I}_j^l   \right |}{   mean( \mathbf{I}^l) },
\end{equation}
where  $mean( \mathbf{I}^l)$ is the mean of the set $\mathbf{I}^l = \{ |\mathbf{I}_i^l\cap \mathbf{I}_j^l|: m_i, m_j \in \mathbf{M}, |\mathbf{I}_i^l\cap \mathbf{I}_j^l| >0 \}$ .

\para{Leveraging side information.}  
If the side information describing additional properties of the user-object interaction is available, we want to include the information in the detection.  For example, \cite{COPYCATCH} reports that the time feature is essential for fraud detection.  To do so, we can augment an edge $\varepsilon_{ij}$ both in $G$ and $G^l$ using the following attribute tuple:  
\begin{displaymath}
\varepsilon_{ij} = \left ( n_i \text{,} \  Attr_1 \text{,} \  Attr_2 \cdots \right ),
\end{displaymath}
where $Attr_i$ can be a timestamp, star-rating, etc. We can append as many attributes as we need into the tuple and combine the synchronized behavior into a single score $C$. We give the following simple example.
\smallskip

\textbf{Example 1:} In a collection of reviews on Amazon, a review action $(n_i, m_j, time_1, IP_1)$ indicates that a user $n_i$ reviewed product $m_j$ at the $time_1$ on $IP_1$. Then, we use $\varepsilon_{ij}$ to denote the review action, $\varepsilon_{ij} =(n_i, time_1, IP_1)$, and we discard $m_i$ for the comparisons in Eq.\ref{eq:c1} and Eq.\ref{eq:c2}.

%

\para{Approximate comparisons.  }
Furthermore, we use a customizable $\tilde{=}$ operator to the set intersection and set union in Eq.\ref{eq:c1} and Eq.\ref{eq:c2}. For example, considering two edge-attribute tuples $\varepsilon_{13}=(n_1, time_1)$ and $\varepsilon_{14} = (n_1 , time_2)$ and let $\Delta$ denote a time range, then $\varepsilon_1 \ \tilde{=} \ \varepsilon_2$ if $time_1 - time_2 < \Delta$.  
To make the computation fast,  we quantize timestamps (e.g., hours) and use $=$ operator.

\para{Reducing the C-score computation complexity. } In the worst case, it takes $O(|\mathbf{M}|^2)$ to compute $C_{ij}, \forall (m_i, m_j) \in \mathbf{M}$, during the OSG construction.

In practice, we only need to compute the object pair $(m_i, m_j)$ with positive \cij. 
We use the \emph{key-value} approach to compute the $S$-score of $C$-score, described in Algorithm \ref{alg:b} (We use the similar method to compute the $S^l$-score ).

\begin{algorithm}[h]
\caption{Building OSG}\label{alg:b}
\begin{algorithmic}[1]
\Require $Dict$
\Ensure $S$-score
\For {each $n_i$ in $\mathbf{N}$}
\State $Dict[n_i] = \{m_j: m_j \in \mathbf{M}, (n_i, m_j)\in E\}$
\EndFor
\For {each $m_i, m_j$ in  $Dict[n_i]$}  $|\mathbf{I}_i \cap \mathbf{I}_j| \leftarrow 0$
\EndFor
\For {each $n_i$ in $\mathbf{N}$}
\For {each $m_i, m_j$ in  $Dict[n_i]$}
\State $|\mathbf{I}_i \cap \mathbf{I}_j| \leftarrow  |\mathbf{I}_i \cap \mathbf{I}_j|+1$ \# because $m_i,m_j$ must share $n_i$.
\EndFor 
\EndFor
\For {each $m_i, m_j$ in $\mathbf{M}$ }
\If {$|\mathbf{I}_i \cap \mathbf{I}_j|$>0}
\State $|\mathbf{I}_i \cup \mathbf{I}_j| \leftarrow  Deg(m_i) + Deg(m_j) - |\mathbf{I}_i \cap \mathbf{I}_j|$ \# $Deg(m_i)$ denotes the in-degree of $m_i$
\State $S_{ij} \leftarrow |\mathbf{I}_i \cap \mathbf{I}_j| / |\mathbf{I}_i \cup \mathbf{I}_j|$
\EndIf
\EndFor
\end{algorithmic}  
\end{algorithm}


Naively, it takes $O(|E|)$ to find all \emph{key-value} pairs (lines 1-2) and takes $O(|\mathbf{E}|)$ to build \obg (lines 4-10). However, we expect \obg to be sparse because an object only has positive $C$-scores with a very small subset of objects in the OSG. Empirically, we evaluate the edge density in several datasets and find the edge density quite low in all cases.  Section~\ref{sec:exp_scale} provides more details. 


Furthermore, due to the Zipf-law, in many real datasets, there are a few objects with extremely high in-degrees in the bipartite graph. For example, a celebrity on Twitter (or a popular store on Amazon) has a vast number of followers (or customers). In our preprocessing step, we delete these nodes and their incoming edges, as the most popular objects are usually not fraudulent.  
This preprocessing significantly reduces $|E|$ and $|\mathbf{E}|$, and thus the overall computation time for OSG construction.

\begin{figure}[tb]
\includegraphics[width = \columnwidth]{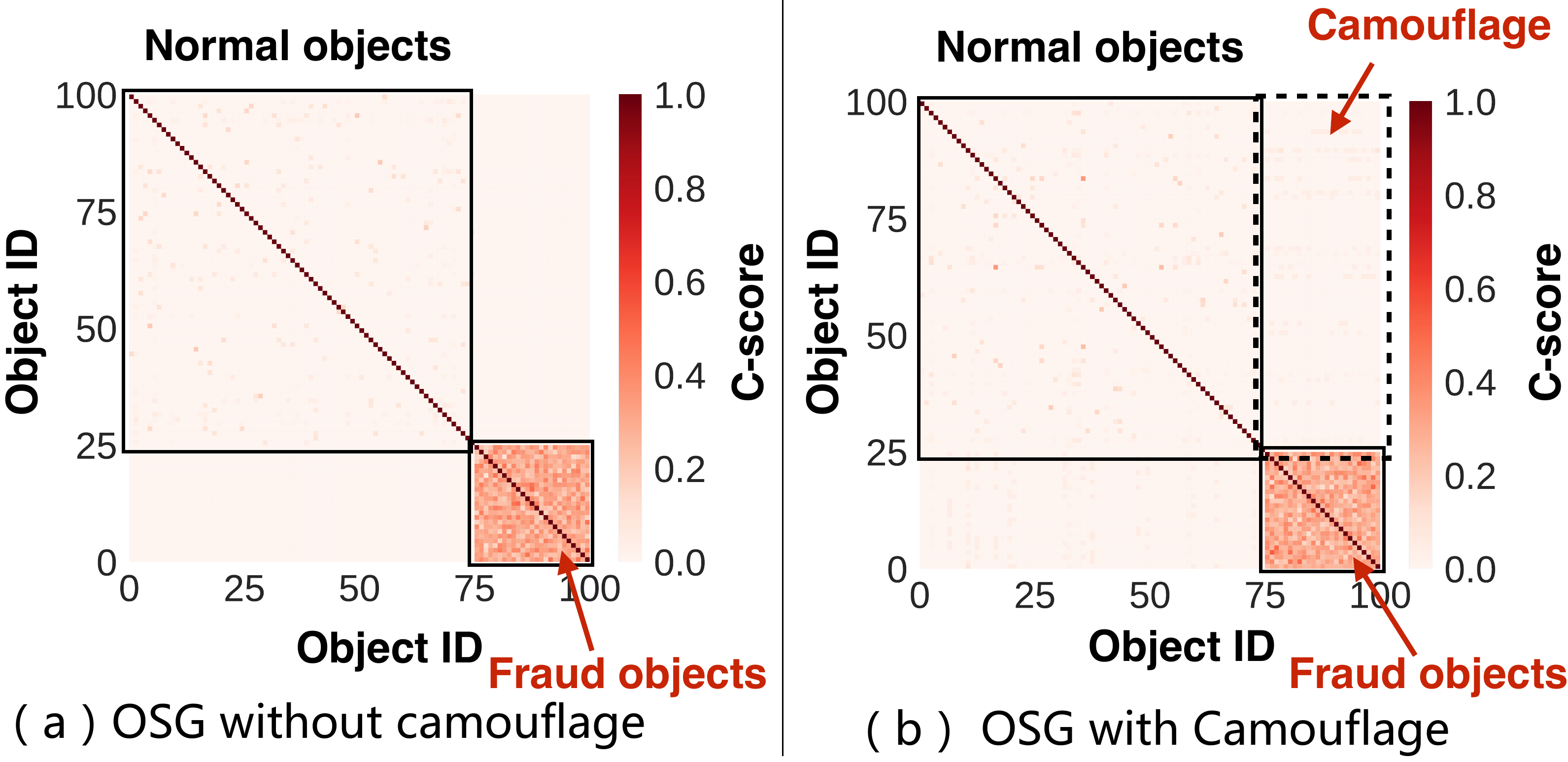}
\centering
\caption{(a) and (b) are the OSG matrixes consisting of 100 objects. ID 1-75 are normal objects sampled from the real-world dataset \cite{amazon_data}. ID 76-100 are fraud objects sampled from an injected fraud group which consists of 100 users and 50 objects (edge density within the group is 0.066). (a) Fraud objects form a dense region in OSG (edge density is 1.0).  (2) We add camouflage edges which are 75 percent of fraudulent edges. Then the dense region formed by the fraud objects does not change, and camouflage edges only produce tiny edge weights in camouflage zone.} \label{fig:cscore}
\vspace{-2.0em}
\end{figure}

\vspace{-1.0em}
\subsection{Stage II: Clustering on OSG (LPA-TK)}

We propose an algorithm, Label Propagation Algorithm based on Top-K Neighbors on Weighted Graph (LPA-TK), to cluster nodes of OSG into groups in face of camouflage.  The algorithm is inspired by LPA\cite{SEMILPA,raghavan2007near} that has proven effective in detecting communities with dense connectivity structure, while LPA only works on the unweighted graph and does not resist the noise/camouflage edges.

LPA-TK takes the OSG $\mathbf{G}$ as input and outputs multiple groups of similar objects, based on the similarity. Algorithms \ref{alg:main}-\ref{alg:tk} describe LPA-TK.

\begin{algorithm}[h]
\caption{LPA-TK} \label{alg:main}
\begin{algorithmic}[1]
\Require  $\mathbf{G = (M, E)}$
\Ensure $\hat{\mathcal{M}}$s
\For {each $m_i \in \mathbf{M}$} $L_i^0 = m_i$  \# Initialize labels  
\EndFor 
\State $\{\mathcal{M}_1, ..., \mathcal{M}_{\delta}\} \leftarrow \mathbf{M}$ \# Color partition: no two adjacent vertices share the same color. 
\State $t \leftarrow 0$,   $\hat{\mathcal{M}}$s$ \leftarrow \emptyset$
\While {the stop criterion is not met}
\State{$t \leftarrow t + 1 $ } \# The $t$-th iteration
\For {$\mathcal{M} \in \{\mathcal{M}_1, ..., \mathcal{M}_{\delta}\}$}
\For {each $m_i \in \mathcal{M}$}
\State $L_i^t = f(m_i, t)$ \# Algorithm 2 defines $f$ 
\EndFor
\EndFor
\EndWhile
\For {each $l \in \{L_i: m_i \in \mathbf{M}\}$}
\State $\hat{\mathcal{M}}$s $\leftarrow \hat{\mathcal{M}}$s $+ \{ \{m_i: m_i \in \mathbf{M}, L_i = l\} \}$
\EndFor
\State \Return $\hat{\mathcal{M}}$s
\end{algorithmic}  
\end{algorithm}
\paragraph{Initialization (Line 1-3)} First, we assign each node in OSG a unique label. Second, we color all nodes so that no adjacent nodes share the same color. The coloring process is efficient and parallelizable, which takes only $O(deg(G))$ synchronous parallel steps \cite{Barenb}. And the number of colors, denoted by $\delta$,  is upper bounded by $deg(G)+1$, where $deg(\mathbf{G})$ denotes the maximum degree over all nodes in $\mathbf{G}$.  

\paragraph{Iterations (Line 4-8)} In the $t$-th iteration, each node $m_i$ updates its label based on the labels of its neighbors (we leave the update criterion to Algorithm \ref{alg:tk}). Since the update of a node's label is only based on its neighbors, we can simultaneously update all the nodes sharing the same color. Thus, we need at most $\delta$ updates per iteration. The iteration continues until it meets the \emph{stop condition}:
\begin{displaymath}
\begin{split}
\forall  m_i \in M: &1) \ L_i^t = L_i^{t-1}  \\
\text{or} \ &2) \ L_i^t \neq L_i^{t-1}, \text{due to a} \  \emph{tie}
\end{split} 
\end{displaymath}
where $L_i^t$ is the label of $m_i$ in the $t$-th iteration, and \emph{tie} represents a condition that $L_i^t$ changes because $f$ returns more than one label choices (line 8). 

\paragraph{Return Groups (Line 9-11)} After the iteration terminates, we partition nodes sharing a same final label into a group.

The key difference of LPA-TK from the original algorithm LPA\cite{SEMILPA} is the designing of update criterion $f$. We consider the three choices of $f$.

\para{[Update Criterion: Sum].}
Obviously, it is significant to design $f$ that determines the final results. Based on the update criterion in \cite{SEMILPA} that only works on unweighted graphs, we first define $f$ as the following form:

\begin{equation} \label{equ:sum}
f(m_i, t)  = \mathop{\arg\max}_{l \in \{L_j^t: m_j \in \mathbf{H}(m_i)\}}  \sum_{m_j \in \mathbf{H}(m_i)} C_{ij} h(L_j^t, l),
\end{equation}
where $\mathbf{H}(m)$ is the set of neighbors of $m_i$ and $h(L_j^t, l)$ is an indicator function:
\begin{displaymath}
h(L_j^t, l) = 
\begin{cases}
1 &  \text{if} \ L_j^t == l.\\
0 & \text{otherwise}.
\end{cases}
\end{displaymath}
According to Eq.\ref{equ:sum}, the label of $m_i$ is determined by the sum of edge weights of each distinct label among its neighbors.
Unfortunately, the results of clustering deteriorate as the camouflage edges increase. Fig.\ref{fig:case}(a) gives an intuitive example. 

\para{[Update Criterian: Max].} To minimize the influence of camouflage, we propose another form of $f$:
\begin{equation} \label{equ:max}
f(m_i, t)  = \mathop{\arg\max}_{l \in \{L_j^t: m_j \in \mathbf{H}(m_i)\}} h(l, m_i) 
\end{equation}
where $\mathbf{H}(m_i)$ is the set of neighbors of $m_i$ and
\begin{displaymath}
 h(l, m_i) = \max\{C_{ij}: m_j \in \mathbf{H}(m_i), L_j^t = l\}
\end{displaymath}

Based on Eq.\ref{equ:max}, the label of $m_i$ is determined by the maximal edge weight of each distinct label among its neighbors. Although Eq.\ref{equ:max} can eliminate the influence of camouflage because the most similar neighbor of a fraud object should also be fraudulent, the result of clustering is not well and a group of objects is often divided into multiple parts.  Fig.\ref{fig:case}(b) gives an example.
\begin{figure}[h]
\includegraphics[width = 0.8 \columnwidth]{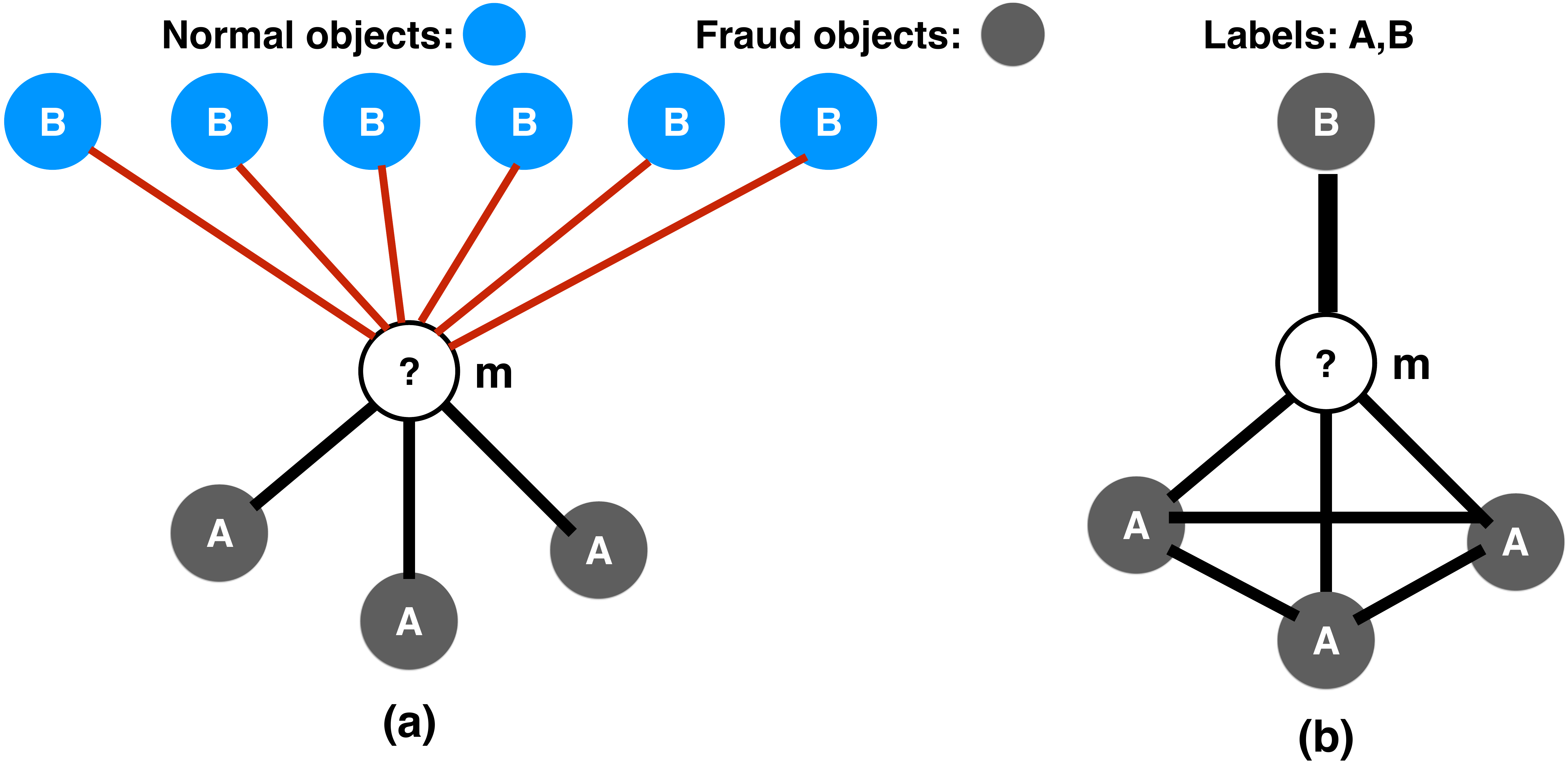}
\vspace{-1.0em}
\centering
\caption{(a) and (b) describe two possible states when determine the label of a fraud object $m$, in which the thickness of an edge indicates the $C$-score. In (a), normal objects labeled as `B' are connected with $m$ because of camouflage edges. According to Eq.\ref{equ:sum}, $m$ will be labeled as `B', while we expect it to be `A'.  In (b), according to Eq.\ref{equ:max}, $m$ will be labeled as `B', while we expect it to be `A'.} \label{fig:case}
\vspace{-1.5em}
\end{figure}

\para{[Update Criterion: Top K].}
Based on these considerations, we propose our final form of $f$, which can eliminate the influence of camouflage and keep ideal clustering results, shown in the Algorithm \ref{alg:tk}.
\begin{algorithm}[h]
\caption{$f(m_i,t)$} \label{alg:tk}
\begin{algorithmic}[1]
\Require$K$
\For {each $l \in \{L_j^t: m_j \in \mathbf{H}(m_i)\}$} \# $\mathbf{H}(m_i)$ is the set of neighbors of $m_i$
\State $h[l] \leftarrow 0$ 
\State $\mathbf{C} \leftarrow \{C_{ij}: m_j \in \mathbf{H}(m_i), L_j^t = l\}$ 
\While {$K > 0$}
\If {$|\mathbf{C}| \geq K$}
\State $h[l] \leftarrow h[l] + \max \mathbf{C}$
\State remove $\max\mathbf{C}$ from $\mathbf{C}$
\EndIf
\State $K \leftarrow K - 1$
\EndWhile
\EndFor
\State \Return $l$ if $h[l]$ is the maximum in $h$ 
\end{algorithmic}  
\end{algorithm}

In Algorithm 2, the label of $m_i$ is determined by the sum of Top-K maximal edge weights of each distinct label among its neighbors. Note that computing the sum of Top-K maximal edge weights (line 4-8) can be optimized to $O(|\mathbf{C}|)$, which is as same as the time complexity of update criteria \ref{equ:sum} and \ref{equ:max}.
  
Empirically, we set $K$ as a small integer (e.g., we set $K = 3$ in our experiments). Not only does LPA-TK resist camouflage (because camouflage edges do not change a fraud object's Top-K most similar neighbors), but also has good clustering performance (eliminate the probability that its label determined by one certain neighbors). 
In Fig.\ref{fig:case}(a) $m$ will be labeled as `A', and in Fig.\ref{fig:case}(b), $m$ will be labeled as `A', using LPA-TK.


The algorithm is deterministic: it always generates the same graph partitions whenever it starts with the same initial node labels. Furthermore, the algorithm converges provably. We formally prove its convergence with the following theorem:

\begin{theorem}
Given a graph $\mathbf{G} = (\mathbf{M}, \mathbf{E})$, $ \forall (m_i, m_j) \in \mathbf{M}$ and \cij $\in \mathbf{E}$, the algorithm\ref{alg:main} uses the updating criterion Algorithm \ref{alg:tk} and the stop condition. Then the algorithm \ref{alg:main} converges.
\end{theorem}

\begin{proof}
Let $f(t)$ be the number of monochromatic edges of \obg at $t$-th iteration step, and $f(t) \leq |\mathbf{E}|$. In the $t$-th step, at least one vertex $i$ changes its label if it does not meet the stop condition. This indicates that $f(t)$ strictly increases during step $t$, i.e. $f(t) > f(t-1)$. Thus the number of iterations is upper bounded by $|\mathbf{E}|$.   
\end{proof}
\smallskip

\vspace{-1.0em}
\subsection{Stage III: Spot Suspicious Groups ($\mathcal{F}$)}\label{sec:detectuser} 

After generating all candidate groups, in this section, 
we propose an interpretable suspiciousness metric $\mathcal{F}$ to score each group and find top suspicious groups. Given a fraud group $\mathcal{M} \in \hat{\mathcal{M}}$s (returned by LPA-TK), let $\mathcal{G}$ be the subgraph of $\mathbf{G}$ induced by $\mathcal{M}$,  $\mathcal{G = (M, E)}$. Then $\mathcal{F}$ follows the form:
\begin{equation}
\mathcal{F_G}= \frac{\sum_{(m_i, m_j) \in \mathcal{E}}C_{ij} \cdot  \sum_{(m_i, m_j) \in \mathcal{E}}  |\mathbf{I}_i \cap \mathbf{I}_j|}{|\mathcal{M}|(|\mathcal{M}|-1)^2} ,
\end{equation}
where $\mathcal{F_G} = \mathcal{F_G}^1  \cdot  \mathcal{F_G}^2 \cdot |\mathcal{M}|$,
\begin{displaymath}
\mathcal{F_G}^1 = \frac{\sum_{ (m_i, m_j) \in \mathcal{E}} C_{ij}}{|\mathcal{M}|(|\mathcal{M}|-1)} 
\end{displaymath}
and 
\begin{displaymath}
\mathcal{F_G}^2 = \frac{\sum_{ (m_i, m_j) \in \mathcal{E}} \left | \mathbf{I}_i \cap \mathbf{I}_j \right | }{|\mathcal{M}|(|\mathcal{M}|-1)}. 
\end{displaymath}
Intuitively, $\mathcal{F_G}^1$ is the average value of $C$-score on all edges of \subg, $\mathcal{F_G}^2$ is the average number of edges pointed from same users on all object pairs of \subg. 

The advantage of $\mathcal{F}$-score is that the score obeys the following good properties including axioms proposed in \cite{CROSSSPOT} that all good algorithms should have. First, we present a well-known metric,  \textbf{edge density} denoted by $\rho_{edge} = \frac{|\mathcal{E}|}{|\mathcal{M}|(|\mathcal{M}|-1)}$. And we use `$\uparrow$', `$\downarrow$' and
`=' to represent `increase', `decrease', and `not change'.

\begin{enumerate}[(i)]
\item AXIOM 1. \emph{[Object Size].} Keeping $\rho_{edge}$, $C_{ij}$, and $|\mathbf{I}_i \cap \mathbf{I}_j|$ fixed, a larger \subg is more suspicious than one with a smaller size.
\begin{displaymath}
(\mathcal{F_G}^1 =)  \wedge  (\mathcal{F_G}^2 =) \wedge (|\mathcal{M}|\uparrow ) \ \Rightarrow  \ \mathcal{F_G}\uparrow
\end{displaymath}

\item AXIOM 2. \emph{[Object Similarity].} Keeping $\rho_{edge}$, $|\mathbf{I}_i \cap \mathbf{I}_j|$, and $|\mathcal{M}|$ fixed, a \subg with more similar object pairs is more suspicious. 
\begin{displaymath}
 C_{ij} \uparrow \ \Rightarrow  (\mathcal{F_G}^1 \uparrow) \wedge  (\mathcal{F_G}^2 =) \ \Rightarrow  \ \mathcal{F_G}\uparrow
\end{displaymath}

\item AXIOM 3. \emph{[User Size].} Keeping $\rho_{edge}$, $C_{ij}$, and $|\mathcal{M}|$  fixed,  a fraud object group (\subg) connected with more fraud users is more suspicious.
\begin{displaymath}
|\mathbf{I}_i \cap \mathbf{I}_j|\uparrow  \ \Rightarrow  (\mathcal{F_G}^1 =) \wedge  (\mathcal{F_G}^2 \uparrow) \ \Rightarrow  \ \mathcal{F_G}\uparrow
\end{displaymath}

\item AXIOM 4. \emph{[Edge Density].} Keeping $C_{ij}$, $|\mathbf{I}_i \cap \mathbf{I}_j|$, and $|\mathcal{M}|$ fixed,  a denser \subg is more suspicious. 
\begin{displaymath}
\rho_{edge} \uparrow \Rightarrow (\mathcal{F_G}^1 \uparrow)  \wedge  (\mathcal{F_G}^2  \uparrow)  \ \Rightarrow  \ \mathcal{F_G}\uparrow
\end{displaymath}

\item AXIOM 5. \emph{[Concentration.]} With the same total suspiciousness, a smaller \subg is more suspicious. We define the total suspiciousness as  $\sum_{(m_i, m_j) \in \mathcal{E}} (C_{ij} + | \mathbf{I}_i \cap \mathbf{I}_j|$). 
\begin{displaymath}
 |\mathcal{M}| \downarrow \ \Rightarrow  \ \mathcal{F_G}\uparrow
\end{displaymath}
\end{enumerate}

Note that naive metrics do not meet all axioms. For example, the edge density is not a good metric because it does not satisfy AXIOM 1-3 and 5. 

Therefore, leveraging $\mathcal{F}$, we can sort groups in descending order of suspiciousness and catch top suspicious groups.

Given suspicious $\mathcal{G}$, we catch fraud users $\mathcal{N}$ from $\mathcal{G}$ comprised of fraud objects $\mathcal{M}$. The approach follows the form: 

\begin{equation}\label{eq:uu}
\mathcal{N} = \bigcup_{\forall m_i,m_j \in \mathcal{M}} H_i \cap H_j
\end{equation}
where $H_i = \{n: \forall n \in \mathbf{N}, (n, m_i) \in E\} $ is the set of users having edges to $m_i$. 

To reduce false alarms in $\mathcal{N}$, we filter out users with low out-degrees in the subgraph induced by $\mathcal{N}$ and $\mathcal{M}$ of $G$, because a normal user may interact with a few fraud objects by accident while it is unlikely that it interacts with many fraud objects.

\subsection{Analysis} \label{sec:analysis}

There are four advantages of \sysname ($C$ + LPA-TK + $\mathcal{F}$) :
\begin{enumerate}
\item \emph{[Camouflage-resistance].}  $C$ + LPA-TK is inherent to resist camouflage (see Theorem \ref{the:cam}). However, for LPA\cite{SEMILPA,raghavan2007near}, its group detection results can be easily destroyed by camouflage (demonstrated in Sec.\ref{sec:lpa}).

\item \emph{[Capturing Loose Synchrony].}  $C$ + LPA-TK + $\mathcal{F}$ focuses on catching loosely synchronized behavior, because its top-K most similar neighbors do not change in OSG for a fraud object. However, The density signal can be decreased significantly by synchrony reduction \cite{FRAUDAR, CROSSSPOT,SPOKEN} (demonstrated in Sec.\ref{exp:syn}).  

\item \emph{[Clustering global similarities].}
Using $C$ + LPA-TK, we cluster nodes based on their similarity. However, 
\cite{FRAUDAR, DCUBE, MZOOM} group nodes based on their degree or density features.    Fig.\ref{fig:points} in Sec.~\ref{sec:exp} shows an intuitive example of the clustering quality.

\item \emph{[Scalability].} 
$C$ + LPA-TK cluster all objects into groups in one run with near-linear time complexity (Sec.\ref{sec:analysis}). Leveraging $\mathcal{F}$, we can obtain Top-$k$ suspicious groups, while \cite{FRAUDAR, DCUBE, MZOOM} only detect a single group per run.  
\end{enumerate}

\para{Time complexity.}
  In the OSG construction stage, it takes $O(|E|+|\mathbf{E}|)$ time, based on the optimization (Sec.\ref{sec:con}). In Stage II, the time cost is the product of the number of iterations and the number of colors, where the former value has been experimentally indicated to grow logarithmically with graph size \cite{SEMILPA} and the latter value is bounded by $deg(\mathbf{G})+1$. In Stage III, it takes $O(|\mathcal{E}|)$ to compute \fscore and catch fraud users of \subg, where $|\mathcal{E}| << |\mathbf{E}|$.  Thus, \sysname has near-linear time complexity.

\para{Capturing loosely synchronized behavior.}  We use a concrete example to show why the algorithm can handle loosely synchronized behaviors.  

Consider a fraud group with 100 fraud users and 50 fraud objects, and each fraud user creates 30 edges to random fraud objects. Let $\mathcal{G}_{orig}$ denote the induced subgraph induced in the original user-object bipartite graph, 
and let $\mathcal{G}_{OSG}$ denote the subgraph formed by fraud objects in OSG. We compute the edge density $\rho_{edge}$ and $\mathcal{F_G}^1$ in Eq.(6) for both $\mathcal{G}_{orig}$ and $\mathcal{G}_{OSG}$.  We have
\begin{displaymath}
\rho_{edge}(\mathcal{G}_{orig}) = 0.134 \ \text{VS} \ \rho_{edge}(\mathcal{G}_{OSG}) = 1 \wedge \mathcal{F}^1_{\mathcal{G}_{OSG}} = 0.506.
\end{displaymath}

Obviously, the subgraph in OSG is much denser than the original bipartite graph. Then, let us reduce the synchrony of fraud group by doubling the number of fraud users and keep the same number of edges. Then we have
\begin{displaymath}
\rho_{edge}(\mathcal{G}_{orig}) = 0.049 \ \text{VS} \ \rho_{edge}(\mathcal{G}_{OSG}) = 1 \wedge \mathcal{F}^1_{\mathcal{G}_{OSG}} = 0.251.
\end{displaymath}
It shows that $\mathcal{G}_{OSG}$ is affected slightly by the reduction of synchrony, compared to $\mathcal{G}_{orig}$. Furthermore, as normal users hardly exhibit synchronized behavior, the \cscore of normal object pairs are close to zero. Thus, \sysname ($C$ + LPA-TK + $\mathcal{F}$) is inherently more effective than approaches relying on density \cite{FRAUDAR, SPOKEN, FBOX}.

\para{Camouflage-resistance. } 
\sysname is robust to resist different types of camouflage. There are two reasons. First, the $\mathcal{F}$-scores of subgraphs induced by fraud object groups do not decrease while adding camouflage edges. Formally, we give the following theorem.  

\begin{theorem} \label{the:cam} 
 Let $\mathcal{G}$ denote a subgraph induced by fraud objects $\mathcal{M}$ , and $\mathcal{N}$ denote the fraud users. $\mathcal{M}$ and $\mathcal{N}$ are from a single fraud group. \subg does not change when users in $\mathcal{N}$ add camouflage edges to non-fraud objects.
 \end{theorem}
 
\begin{proof}
Let $m_i$ and $m_j$ denote two fraud objects, $(m_i, m_j) \in \mathcal{M}$. Camouflage only introduces edges between $\mathcal{N}$ and normal objects. It does not add or remove edges pointing to $\mathcal{M}$, which demonstrates that $\mathbf{I}_i$ and $\mathbf{I}_j$ in Eq.(1) do not change. Thus $C_{ij} \in \mathcal{G}$ does not change, $\forall (m_i, m_j) \in \mathcal{M}$. 
\end{proof}
\smallskip

Second, in OSG, a camouflage edge between a fraud user and a normal object only produces a quite small value of $C$-score due to the denominator of Eq. (\ref{equ:cscore}). Fig. \ref{fig:cscore} (b) provides a typical case.  For a fraud user, this indicates that  camouflage edges do not change its the most top-$K$ similar neighbors. Thus, the subgraphs induced by fraud groups can be effectively detected by LPA-TK.

\para{Effectiveness of the semi-supervised mode.}
Given a subset of labeled fraud users, \sysname switches to the semi-supervised mode. 
Because of the design of $C$-score, the partially labeled data does enhance the similarities between fraud objects in a group and increase the density of induced subgraph on OSG. Thus, unsurprisingly, LPA-TK will more accurately cluster fraud objects into groups. The experiments in Section \ref{exp:syn} show the fact.
%
%
%

%

\section{Experiments and results} 
\label{sec:exp}

We want to answer the following questions in the evaluation:

\begin{enumerate}[(1)]
\item	How does \sysname handle loosely synchronized behavior?
\item Is \sysname robust with different camouflage{?}
\item Does the semi-supervised mode improve the performance?
\item Is \sysname scalable to large real-world datasets?
\end{enumerate}
Table\ref{tab:details} gives the details of datasets used in the paper.

\begin{table}[h]
	\caption{Datasets used in experiments. }\label{tab:details}
	\centering
	\begin{tabular}{cc|cccc}
		\toprule
		 datasets & edges & datasets & edges \\
		 \midrule 
		 AmazonOffice\cite{amazon_data} &53K  & YelpChi\cite{YELP} &67K    \\
		 AmazonBaby\cite{amazon_data}&160K  & YelpNYC\cite{YELP}& 359K   \\
		 AmazonTools\cite{amazon_data}& 134K & YelpZip\cite{YELP}& 1.14M  \\
		 AmazonFood\cite{amazon_data} &1.3 M & DARPA\cite{DARPA} & 4.55M  \\
		 AmazonVideo\cite{amazon_data} & 583K & Registration  & 26k\\
		 AmazonPet\cite{amazon_data} & 157K & &    \\
		  \bottomrule
	\end{tabular}
	\vspace{-1em}
\end{table}

\para{Implementation and existing methods in comparison.} 
We implemented \sysname by Python and we run all experiments on a server with two 10-core $2.2$ GHz Intel Xeon E5 CPUs and $64$ GB memory.
We compared \sysname with the following three state-of-the-art methods that focus on synchronized behavior with application to fraud detection. 

\begin{enumerate}[(1)]
\item 	\emph{Fraudar}\cite{FRAUDAR} finds the subgraph with the maximal average degree in the bipartite graph using an approximated greedy algorithm.  It is designed to be camouflage-resistance. 
\item \emph{CatchSync}\cite{CATCHSYNC} specializes in catching rare connectivity structures of fraud groups that exhibit the synchronized behavior, it proposes the synchronicity and normality features based on the degree and HITS score of the user.
\item \emph{CrossSpot}\cite{CROSSSPOT} detects the dense blocks which maximize the suspiciousness metric in the multi-dimensional dataset.
\end{enumerate}

We did our best to fine-tune the hyper-parameters to achieve their best performances. For CrossSpot, we set the random initialization seeds as 400, 500 and 600, and chose the one with the maximal F1-score. Fraudar detects the subgraph with the maximal average degree and multiple subgraphs by deleting previously detected nodes and their edges. For all methods, we test the performance according to the rank of the suspiciousness scores.
We compared the performance using the standard metric, F1 score (the harmonic mean of precision and recall) across all algorithms. 

\para{FraudTrap and FraudTrap+.}  We run \sysname in two modes. The unsupervised mode (\sysname) and the semi-supervised mode (\sysnamep) assuming  5\% fraud users are randomly labeled. And in all experiments, we set $K = 3$ for LPA-TK. In the experiments regarding [Amazon] datasets, assume $\mathcal{M}$ is a fraud object group returned by \sysname, $\mathcal{N}$ is a fraud user group returned by Eq.\ref{eq:uu}. Then we filtered out $n$ if the out-degree of $n$ is less than 3 in the subgraph induced by $\mathcal{(N,M)}$ of $G$, $\forall n \in \mathcal{N}$.

\para{Fraud Group Formulation.}
 To simulate the attack models of smart fraudsters, we used the same method as \cite{FRAUDAR,CATCHSYNC} to generate labels: inject fraud groups into Amazon datasets (\textit{[Amazon]} datasets contain six collections of reviews for different types of commodities on Amazon, listed in Table \ref{tab:details}.) . To accurately depict the injection, we formulate the fraud group as the following.

\begin{definition} 
[$\rho$-Synchrony] Given a subgraph $\mathcal{G}(|\mathcal{N}|, |\mathcal{M}|, \rho,  \theta)$ induced by a group $(\mathcal{N,M})$ in $G$ where $\mathcal{N}$ is a set of users, $\mathcal{M}$ is a set of objects. (1) For each $n \in \mathcal{N}$, $ \exists \mathcal{W} \subseteq \mathcal{M}$ where for each $m \in \mathcal{W}$, the edge $(n, m)$ exists. We define $\rho$ as
\begin{displaymath}
\rho = \frac{\overline{| \mathcal{W}|}}{|\mathcal{M}|},
\end{displaymath}
where $\overline{| \mathcal{W}|}$ is the mean for all $|\mathcal{W}|$s. (2) For each $n \in \mathcal{N}$, $\exists \mathcal{W}$ and $\mathcal{W} \cap \mathcal{M} = \emptyset$, where for each $m \in \mathcal{W}$, the edge $(n, m)$ exists. We set $\theta = \overline{| \mathcal{W}|}$, and $\overline{| \mathcal{W}|}$ is the mean for all $|\mathcal{W}|$s.
\end{definition}
Thus, we use  $\mathcal{G}(|\mathcal{N}|, |\mathcal{M}|, \rho,  \theta)$ to represent a fraud group, where $\rho$ represents how loosely its synchronized behavior is and $\theta$ denotes the number of camouflage edges of each user on average. Naturally, $\mathcal{N}$ and $\mathcal{M}$ are labeled as `fraudulent'.
Before we evaluate the performance of \sysname, we first verify the effectiveness of LPA-TK.

\subsection{Performance of LPA-TK} \label{sec:lpa}

We recall that LAP-TK has the best clustering performance and camouflage-resistance. We design this experiment to demonstrate its performance. We injected a fraud group $\mathcal{G}(|\mathcal{N}| =200, |\mathcal{M}| = 50, \rho = 0.3,  \theta)$ into AmazonOffice dataset, where $\rho = 0.3$ indicates that each fraud user of $\mathcal{N}$ reviews 15 fraud objects of $\mathcal{M}$ and $\theta$ represents the number of camouflage edges of each fraud user on average. We varied $\theta$ to specifically examine the resistance to camouflage of each clustering algorithm.  Let $G$ denote the bipartite graph formed by injected AmazonOffice and we built the OSG of $G$ using the method in section \ref{sec:con}, $\mathbf{G}$. We run each algorithm on $\mathbf{G}$ and evaluated the clustering performance and the performance of detecting fraud objects, and we used metric $\mathcal{F}$ to compute suspiciousness scores of detected groups. Note that we only injected one group and thus $\mathcal{M}$ should be clustered into one group. \emph{LPA} denotes the algorithm \cite{SEMILPA} that treats each edge weight equally, \emph{LPA-Sum} denotes the Algorithm \ref{alg:main} + Eq.\ref{equ:sum} and \emph{LPA-Max} denotes the Algorithm \ref{alg:main} + Eq.\ref{equ:max}.

\begin{table}[h]
	\vspace{-1em}
	\caption{Clustering performance on the [Amazon] datasets. `Num' represents `the number of groups $\mathcal{M}$ is divided into. $\theta$ denotes the number of camouflage edges of each fraud user. }\label{tab:lpa}
	\centering
	\begin{tabular}{c|c|c|c|c}
		\toprule
		& $\theta $ = 0 &$\theta $ = 5 &$\theta $ = 10 &$\theta $ =20\\
		\midrule  
		& Num|AUC &Num|AUC & Num|AUC & Num|AUC \\
		\midrule  
		 LPA  & 1 |1.0 & 1 |0.787  &1 | 0.727&  1 |0.731 \\
		 LPA-Sum & 1 | 1.0 & 1 | 1.0 & 1 | 0.787& 1 | 0.761  \\
		 LPA-Max  & 14 | 0.998 & 14 | 0.996 &  13 | 0.998  &10 |0.991 \\
		\midrule  
		 LPA-TK & \textbf{1} |\textbf{1.0}& \textbf{1} |\textbf{1.0} & \textbf{1} | \textbf{0.999} & \textbf{1}| \textbf{0.998}\\
		\bottomrule
	\end{tabular}
	\vspace{-1em}
\end{table}

Table \ref{tab:lpa} presents the clustering performance of each algorithm. In the setup, we expect `AUC' = 1.0 for the perfect performance of detecting fraud objects and `Num' = 1 for the perfect clustering result. Then we have the following observations: (1) without camouflage ($\theta = 0$), LPA has an ideal performance. However, once camouflage is added ($\theta \geq 0$), its performance is destroyed and \emph{LPA clustered all objects into one group} (thus `num' = 1). (2) LPA-Sum shows weak camouflage-resistance, and it performance deteriorates as camouflage edges increases. (3) LAP-Max resists camouflage obviously. However, it divided $\mathcal{M}$ into multiple groups, which is not good to group analysis and inspection. (4) Our algorithm LPA-TK has perfect performance. It clustered all fraud objects into one group and separated the group from legit objects even in face of camouflage. Thus the experiments demonstrate the advantages of LPA-TK.

\subsection{Overall Performance of FraudTrap} \label{exp:syn}

To [Amazon] datasets, we designed two fraud group injection schemes: the first is to examine in detail the performances for detecting loosely synchronized behavior and resisting camouflage; the second is for more general performance evaluation.

\begin{figure}
	\includegraphics[width = \columnwidth]{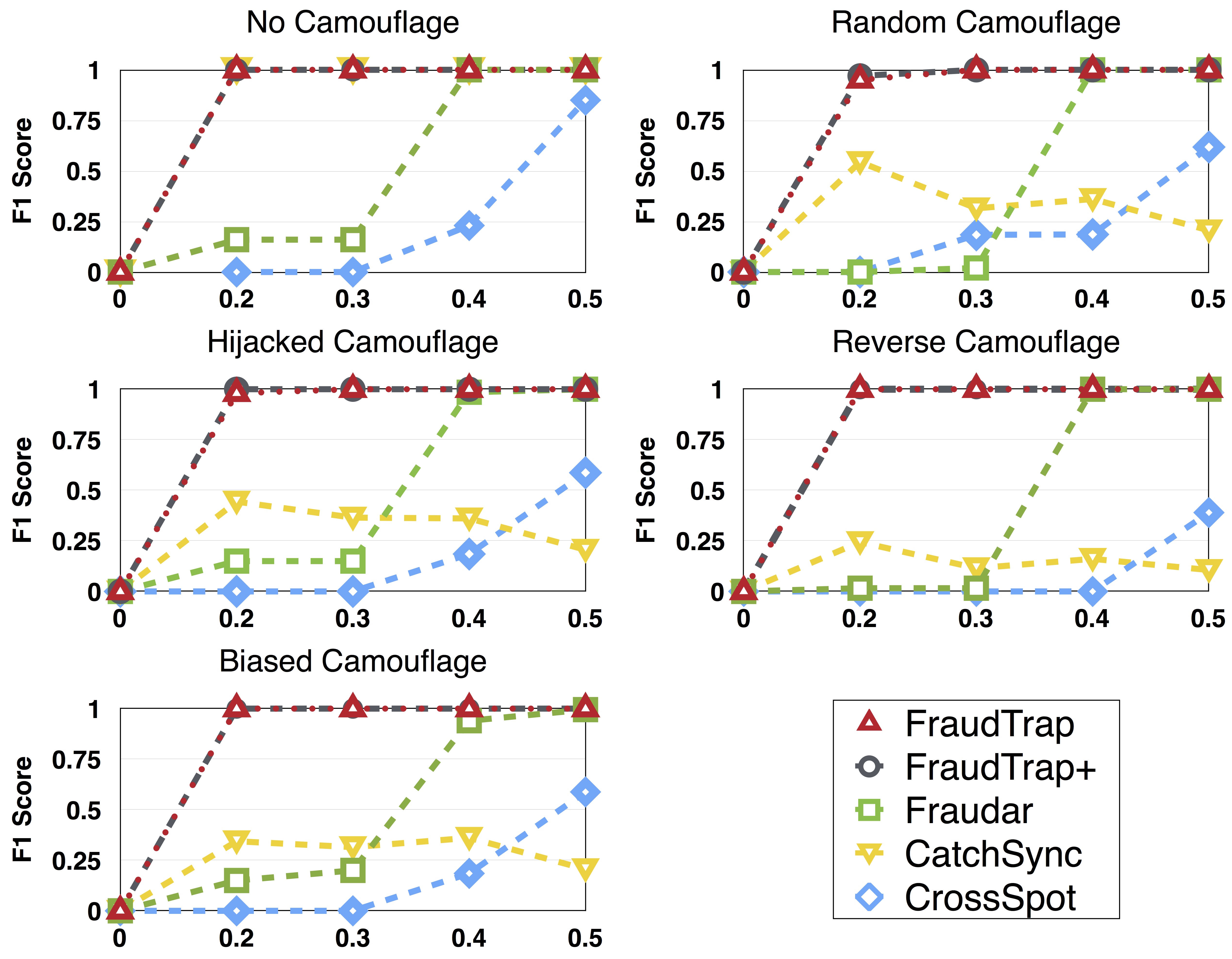}
	\vspace{-1.5em}
	\caption{X-axis: $\rho$. Performance on detecting fraud objects} \label{fig:obj}
\end{figure}

\begin{figure}
	\includegraphics[width = \columnwidth]{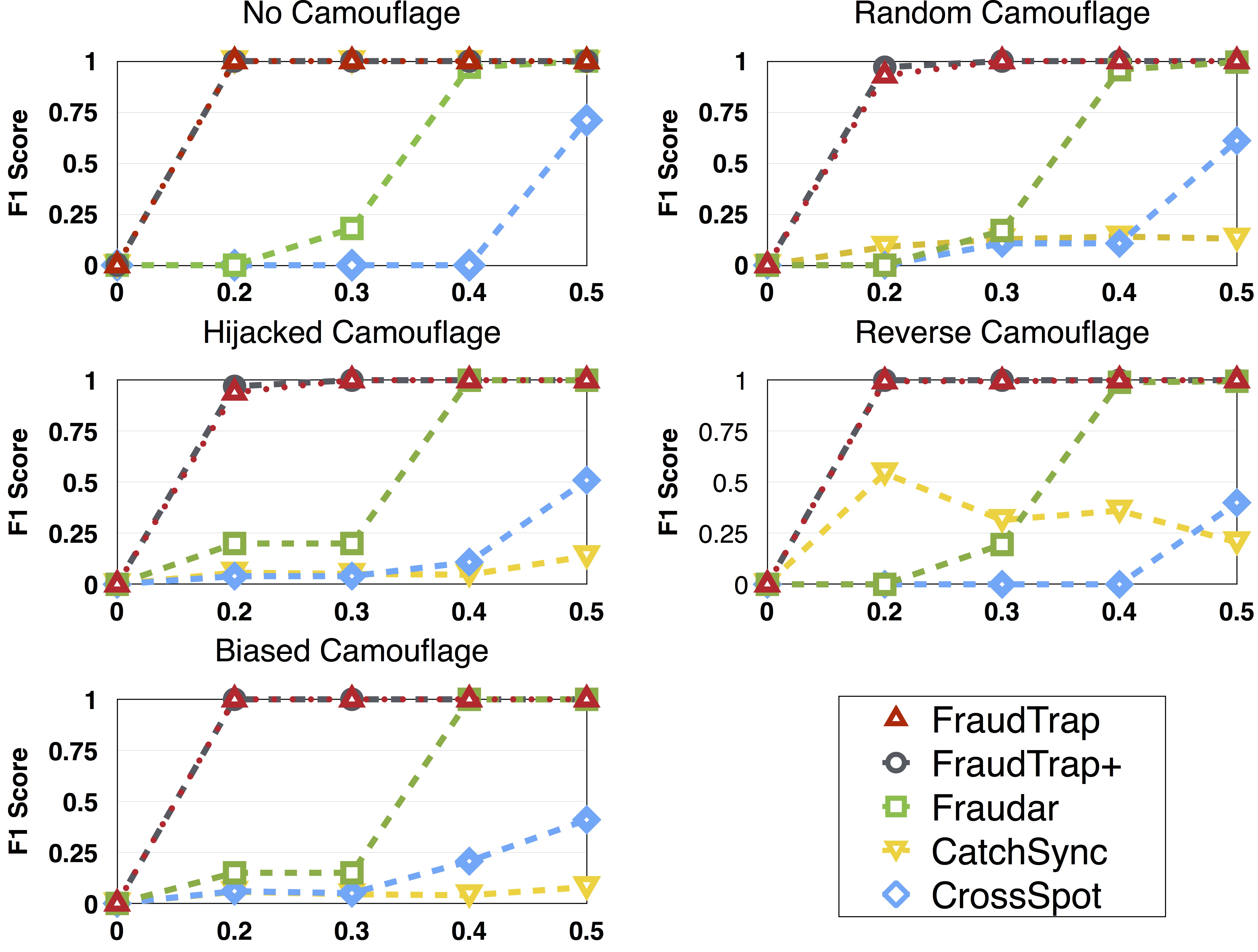}
		\vspace{-1.5em}
	\caption{X-axis: $\rho$. Performance on detecting fraud users} \label{fig:user}
\vspace{-1.5em}
\end{figure}

\para{[Injection Scheme 1].} We chose AmazonOffice as the representative and injected a fraud group into it with varying configurations. We set the fraud group as $\mathcal{G}(|\mathcal{N}| =200, |\mathcal{M}| = 50, \rho,  \theta = |\mathcal{M}| \times \rho)$. We introduced two perturbations according to strategies of smart fraudsters: (1) reduce synchrony by decreasing $\rho$; (2) add camouflage edges obeying $\theta$. In $\mathcal{G}$,  $\theta = |\mathcal{M}| \times \rho$ indicates that the number of a node's camouflage edges is equal to the number of its edges within the fraud group. 
And we used all four types of camouflage as in Sec. \ref{sec:camouflage}.

 Figure \ref{fig:user} and Figure \ref{fig:obj} summarize the performance of detection fraud objects and fraud users with varying $\rho$ respectively, where the X-axises are the synchronization ratio $\rho$ (varying from 0 to 0.5), and the Y-axises are the F1 scores. We have the following observations.
1) Without camouflage and with a high synchronization ratio $\rho$, both \sysname and CatchSync can catch all injected frauds. 
2) At lower $\rho$'s, even without camouflage, the performance of Fraudar decreases significantly, but \sysname maintains its performance.  In fact, even for $\rho = 0.1$, the edge density of the fraud group is only $0.006$, \sysname can still achieve an F1 score of 0.97.  The effects confirm the robustness of our novel approach ($C$ + LPA-TK + $\mathcal{F}$).  
3) Camouflage significantly decreases the performance of CatchSync, but both \sysname and Fraudar are resistant to camouflage.  Not surprisingly, \sysname performs much better when camouflage and loose synchronization exist together.  
4) As shown in Fig.~\ref{fig:user}, without the camouflage and loose synchronization, CatchSync and Fraudar perform perfectly, but their performance degrade quickly when $\rho$ decreases with camouflage.  
5) CrossSpot performs poorly for any $\rho <= 0.5$. 
 
\para{[Injection Scheme 2].} In this experiment, we injected 5 fraud groups $\mathcal{G}(|\mathcal{N}| =200, |\mathcal{M}| = 50, \rho,  \theta = |\mathcal{M}| \times \rho)$ into AmazonOffice, AmazonBaby,  AmazonTools, AmazonFood, AmazonVideo, and AmazonBook, in which $\rho$ was randomly chosen from $[0.2, 0.6]$, respectively. Out of the 5 fraud groups, 1 of them is no camouflage, 4 out of them are augmented with four types of camouflage respectively. The performances are shown in Table \ref{tab:amazon1} and Table \ref{tab:amazon2} with respect to the detection of fraud objects and users. Overall, \sysname is the most robust and accurate across all variations and camouflages. The semi-supervised \sysnamep with a random selection of 5\% labeled fraud users kept or further improve the performance, verifying the conclusion in Section \ref{sec:analysis}.

\vspace{-1em}
\begin{table}[h]
	\caption{Performance(AUC) of detecting fraud objects on the [Amazon] datasets }\label{tab:amazon1}
	\centering
	\begin{tabular}{c|cccc}
		\toprule
		& AmazonOffice & AmazonBaby& AmazonTools \\
		\midrule  
		Fraudar &0.8915 &0.8574 & 0.8764  \\
		 CatchSync & 0.8512 &0.8290& 0.8307  \\
		 CrossSpot &0.8213 & 0.8342 & 0.7923 \\
		\midrule  
		 \sysname & \textbf{0.9987} & \textbf{0.9495} & \textbf{0.9689}\\
		 \sysnamep & \textbf{0.9987} & \textbf{0.9545} & \textbf{0.9675}\\
		\toprule
		&AmazonFood &AmazonVideo & AmazonBook \\
		\midrule  
		Fraudar & 0.6915 & 0.7361 & 0.8923  \\
		 CatchSync & 0.7612 &0.7990& 0.7634  \\
		 CrossSpot & 0.7732 & 0.7854 & 0.8324 \\
		\midrule  
		\sysname & \textbf{0.8458} & \textbf{0.8651} & \textbf{0.9534}\\
		\sysnamep & \textbf{0.8758} & \textbf{0.8951} & \textbf{0.9644}\\
		\bottomrule
	\end{tabular}
	\vspace{-1em}
\end{table}

\vspace{-1em}

\begin{table}[h]
	\caption{Performance(AUC) of detect fraud users on the [Amazon] datasets }\label{tab:amazon2}
	\centering
	\begin{tabular}{c|cccc}
		\toprule
		& AmazonOffice & AmazonBaby& AmazonTools \\
		\midrule  
		Fraudar &0.9015 &0.8673 & 0.8734  \\
		 CatchSync & 0.8732 &0.8391& 0.8304  \\
		 CrossSpot &0.8113 & 0.8422 & 0.7823 \\
		\midrule  
		 \sysname & \textbf{1} & \textbf{0.9795} & \textbf{0.9796}\\
		 \sysnamep & \textbf{1} & \textbf{0.9845} & \textbf{0.9855}\\
		\toprule
		&AmazonFood &AmazonVideo & AmazonBook \\
		\midrule  
		Fraudar & 0.7213 & 0.7451 & 0.8815  \\
		 CatchSync & 0.7234 &0.8243& 0.7763  \\
		 CrossSpot & 0.7653 & 0.7913 & 0.8532 \\
		\midrule  
		\sysname & \textbf{0.8637} & \textbf{0.8843} & \textbf{0.9572}\\
		\sysnamep & \textbf{0.8818} & \textbf{0.9111} & \textbf{0.9579}\\
		\bottomrule
	\end{tabular}
	\vspace{-1em}
\end{table}

\para{[Yelp]}\cite{YELP}. YelpChi, YelpNYC, and YelpZip are three datasets collected by \cite{YELP2} and \cite{YELP}, which contain a different number of reviews for restaurants on Yelp. Each review includes the \emph{user} who made the review and the \emph{restaurants}. Thus the three datasets can be represented by the bipartite graph $G$ formed by (users, restaurants).
The three datasets all include labels indicating whether each review is fake or not.
Detecting fraudulent users has been studied in \cite{YELP} but used review text information.
In this paper, we give the evaluation of catching fraudulent restaurants which bought fake reviews only using the two features. Intuitively, more reviews a restaurant contains, the more suspicious it is. Therefore, we label a restaurant as ``fraudulent'' if the number of fake reviews it receives exceeds 40 (because legit restaurants also may contain a few fake reviews). Table \ref{tab:yelp} shows the results. \sysname and \sysnamep have the best accuracy on all three datasets.

\vspace{-1em}
\begin{table}[h]
	\caption{Performance(AUC) on the [YELP] datasets}\label{tab:yelp}
	\centering
	\begin{tabular}{l|ccc}
		\toprule
		&YelpChi & YelpNYC & YelpZip\\
		\midrule  
		Fraudar& 0.9905  & 0.8531 & 0.7471 \\
		CatchSync&0.9889 &0.8458 & 0.7779 \\
		CrossSpot & 0.9744 &0.7965 &0.7521\\ 
		\midrule  
		\sysname & \textbf{0.9905}& \textbf{0.8613}& \textbf{0.7793} \\
		\sysnamep & \textbf{0.9905}& \textbf{0.8653}& \textbf{0.7953} \\
		\bottomrule
	\end{tabular}
	\vspace{-1em}
\end{table}

\para{[DARPA]} DARPA\cite{DARPA} was collected by the Cyber Systems and Technology Group in 1998. It is a collection of network connections, some of which are TCP attacks. Each connection contains \emph{source IP} and \emph{destination IP}. Thus the dataset can be modeled as a bipartite graph $G$ formed by (source IPs, destination IPs) and we evaluate the performance of detecting malicious source IPs and destination IPs respectively for each method. The dataset includes labels indicating whether each connection is malicious or not and we labeled an IP as `malicious' if it was involved in a malicious connection. 

Table \ref{tab:darpa} presents the corresponding accuracies. Unfortunately, all baselines have bad performance regarding the detection of malicious IPs. However, \sysname and \sysnamep exhibit near-perfect accuracy. 

\vspace{-1em}
\begin{table}[h]
	\caption{Performance(AUC) on the [DARPA] dataset}\label{tab:darpa}
	\centering
	\begin{tabular}{c|cc}
		\toprule
		Detection of & source IP &  destination IP \\
		\midrule  
        Fraudar & 0.7420 & 0.7298 \\
        CatchSync & 0.8069 & 0.8283 \\
        CrossSpot & 0.7249 & 0.6784 \\
		\midrule  
		\sysname & \textbf{0.9968} & \textbf{0.9920}\\
		\sysnamep & \textbf{0.9968} & \textbf{0.9920}\\
		\bottomrule
	\end{tabular}
	\vspace{-1em}
\end{table}

\begin{figure*}[h]
	\includegraphics[width = \textwidth]{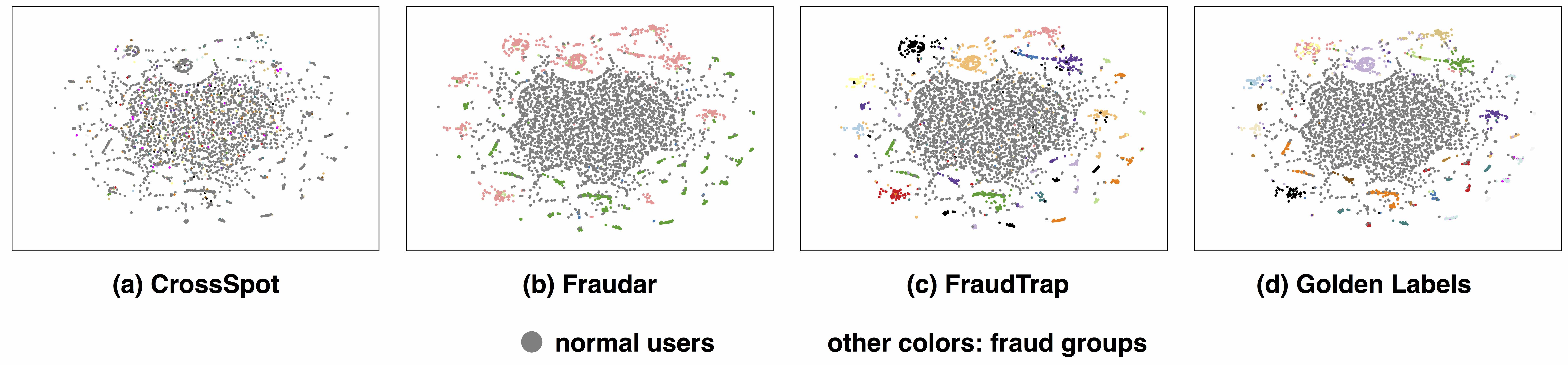}
	\vspace{-2em}
	\caption{The projection of the example dataset using t-SNE on [Registration](See text). } \label{fig:points}
	\vspace{-0.5em}
\end{figure*}

\para{[Registration]} is a real-world user registration dataset with 26k log entries from a large e-commerce website.  Each entry contains \emph{user ID} and two more features, \emph{IP subnet} and \emph{phone-prefix}, and an additional feature \emph{timestamp}.
The dataset includes labels indicating whether each entry (user ID) is fraudulent or not and the labels are obtained by tracking these user accounts for several months, and see if the user conduct fraud after registration and each fraud user has a group ID according to obvious attribute sharing among them. Of all user accounts, 10k are fraudulent. 
Note that the registration records do not contain the user-object interaction.  We can easily adapt it to the \sysname framework assuming that each registered account (``object'')  has many followers identified by a feature value,  IP subnet or phone-prefix (``user'').  Intuitively, an IP subnet can be used in many registrations, and we model it the same as a user follows multiple objects in a social network. Moreover, we use \sysnames to denote the mode of \sysname using the side information \emph{timestamp}, and an edge $\varepsilon = (IP subnet/phone, timestamp)$ in \sysnames.


\begin{table}[h]
	\caption{Performance (AUC) on the [Registration] dataset}\label{tab:f1}
	\centering
	\begin{tabular}{l|c|c}
		\toprule
		  & Feature: IP  & Feature: phone \\
		\midrule  
		 Fraudar   &  0.7543  & 0.8742 \\
		 CatchSync  & 0.7242 & 0.8435  \\
		 CrossSpot   & 0.6976  & 0.8231 \\		
		 \midrule  
		\sysname & \textbf{0.7658}& \textbf{0.8979} \\
		\textsc{\sysname}+&\textbf{0.7826}&\textbf{0.9113}\\
		\sysnames & \textbf{0.7724}&\textbf{0.9215}\\
		\bottomrule
	\end{tabular}
	\vspace{-1em}
	
\end{table}

In our first experiment, we only used the IP subnet feature as the ``user'' side of the bipartite graph. The left half of Table~\ref{tab:f1} summarizes the results.  We have the following observations:
1) \sysname, \sysnamep, and \sysnames outperformed all the other existing methods by a small margin.  Taking a closer look at the detection result, we found that \sysname captured a fraud group of 75 fraud users that all other methods missed.  The group is quite loosely synchronized with the edge density of only 0.14 in the original bipartite graph.  However, having an edge density in 1.0 in OSG makes it highly suspicious in \sysname. 
2) \sysnamep performed better than the unsupervised version, even with only 5\% of the fraud labels.  

In the second experiment, we used phone features as the ``user'' side of the bipartite graph.  The right column of Table~\ref{tab:f1} summarizes the results.  The key observations are: 1) \sysname, \sysnamep, and \sysnames still outperformed other methods. Other baselines have lower performance because they worked poorly on a group with 125 false positives (and 75 true positives).  This is because they are based on edge density on the bipartite graph only, and the groups' edge density is too big enough to distinguish this group from the normal. 

As an additional benefit, \sysname can provide insights on the grouping of fraud users/objects by their similarity.  Fig.~\ref{fig:points} plots a projection of the data onto a 2D space using t-SNE\cite{Van2017Visualizing} on [Registration].  We labeled the fraud groups and normal groups according to the $\mathcal{F}$-score ranking.  We plotted the users in the same group with the same color. We expect that points in the similar groups are clustered together.  We observe that the results from \sysname are much better than Fraudar and CrossSpot, since the users with the same color cluster better, which is very similar to the clustering result of labels.


\subsection{Scalability}\label{sec:exp_scale}

\para{Sparsity of OSG edges.  }  All the three datasets above have low edge densities.  In fact, we also studied several public datasets by the construction of the OSG and the computation of the edge density.  For example,  three datasets in \cite{amazon_data}  and one dataset in \cite{Leskovec2010Signed} have edge densities of 0.0027, 0.0027, 0.0028, 0.0013 respectively.  With this density, the time and space complexities of \sysname are both near-linear to the number of edges in the graph.
\begin{figure}[h]
\includegraphics[width = 0.8\columnwidth, height = 0.5\columnwidth]{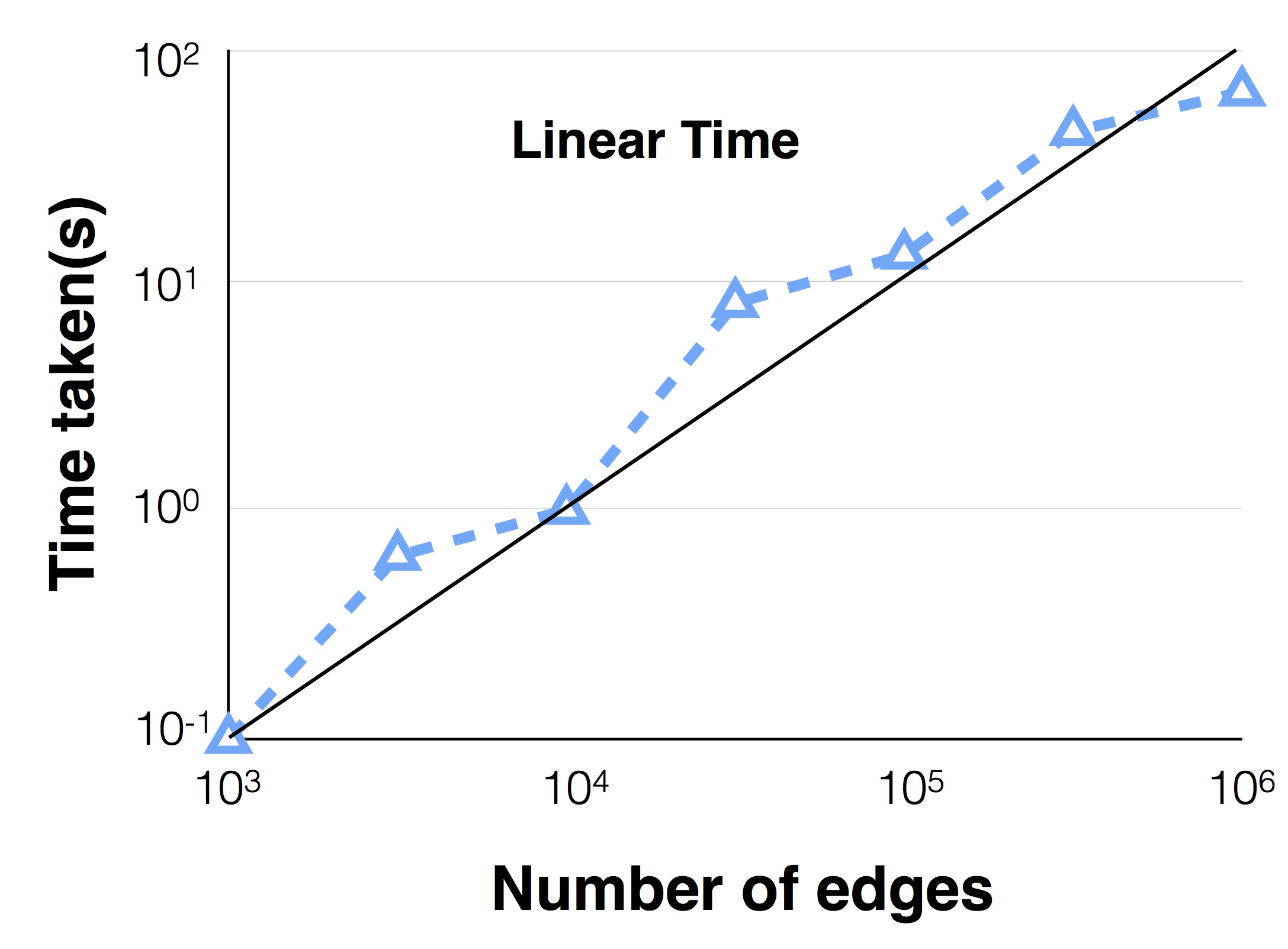}
\vspace{-1em}
\caption{\sysname has the near-linear time complexity: the curve (blue) show the running time of \sysname, compared with a linear function (Black)} \label{fig:time}
\vspace{-1em}
\end{figure}

Based on the dataset AmazonFood, we vary the number of edges using downsampling, and verify the running time of \sysname is indeed near-linear, as shown in Fig.~\ref{fig:time}.

\section{Conclusion}

Fraudsters can adapt their behavior to avoid detection.  Specifically, they can reduce their synchronized behaviors and conduct camouflage to destroy the performance of state-of-the-art methods in the literature.  
To solve the challenges,
we propose \sysname to capture the more fundamental similarity among fraud objects, and work on the edge density on the Object Similarity Graph (OSG) instead.  We design \sysname with many practical considerations for the general fraud detection scenario in many applications, such as supporting a mixture of unsupervised and semi-supervised learning modes, as well as multiple features.  We believe the metrics of \sysname are much harder for fraudsters to manipulate.

%

\bibliographystyle{ACM-Reference-Format}
\bibliography{Ref}
\end{document}